\documentclass[letterpaper]{article} %DO NOT CHANGE THIS
\usepackage{aaai19}  %Required
\usepackage{times}  %Required
\usepackage{helvet}  %Required
\usepackage{courier}  %Required
\usepackage{url}  %Required
\usepackage{graphicx}  %Required
\usepackage{microtype}
\usepackage{booktabs} % for professional tables
\usepackage{mathtools}
\usepackage{amsfonts}
\usepackage{amsmath,amssymb,amsthm}
\usepackage{mathrsfs}
\usepackage{nicefrac}  % compact symbols for \nicefrac[]{1}{2}, etc.

\usepackage{bm}
\usepackage{bbm}
\usepackage{stmaryrd}
\usepackage{enumerate}
\usepackage{algorithm}
\usepackage{algorithmic}
\usepackage{tablefootnote}    % footnote in table and tabular env
\usepackage{footmisc}   % \footref, refer the same footnote at different places
\usepackage{subcaption} % sub-figures, cannot be used with subfigure package
\usepackage{setspace}   % set space between lines
\usepackage[Symbolsmallscale]{upgreek}
\usepackage{multirow} % for tables
\usepackage{colortbl} % cellcolor
\graphicspath{{fig/}}   % Location of the graphics files

\newtheorem{lemma}{Lemma}

\DeclareMathOperator*{\argmin}{argmin}

\newcommand{\eat}[1]{}

\newcommand{\llb}{\llbracket}
\newcommand{\rrb}{\rrbracket}

\newcommand{\x}{\mathbf{x}}

\newcommand{\w}{\mathbf{w}}

\newcommand{\Z}{\mathbf{Z}}

\newcommand{\R}{\mathbb{R}}

\newcommand{\DCal}{\mathcal{D}}

\newcommand{\UCal}{\mathcal{U}}

\newcommand{\alphabm}{{\bm{\alpha}}}
\newcommand{\betabm}{{\bm{\beta}}}

\newcommand{\bmu}{\bm{\mu}}

\newcommand{\upthetabm}{{\bm{\uptheta}}}

\newcommand{\zero}{\mathbf{0}}

\newcommand{\widebar}[1]{\mkern 1.5mu\overline{\mkern-2.5mu#1\mkern-1.5mu}\mkern 1.5mu}

\newcommand{\eg}{e.g.,\ }
\newcommand{\ie}{i.e.,\ }

\newcommand{\citet}[1]{\citeauthor{#1}\ \shortcite{#1}}
\newcommand{\citep}{\cite}

\begin{document}
\title{Cold-start Playlist Recommendation with Multitask Learning}
\author{Dawei Chen$^{* \dagger}$, Cheng Soon Ong$^{\dagger *}$, Aditya Krishna Menon$^{*}$ \\
$^*$The Australian National University, $^\dagger$Data61, CSIRO, Australia \\
\{dawei.chen, chengsoon.ong, aditya.menon\}@anu.edu.au
}

\maketitle

\begin{abstract}
Playlist recommendation involves producing a set of songs that a user might enjoy.
We investigate this problem in three %different 
cold-start scenarios:
(i) \emph{cold playlists}, 
where we recommend songs to form new personalised playlists for an existing user;
(ii) \emph{cold users},
where we recommend songs to form new playlists for a new user; and
(iii) \emph{cold songs}, where we recommend newly released songs to extend users' existing playlists.
We propose a flexible multitask learning method to deal with all three settings.
The method learns from user-curated playlists,
and encourages songs in a playlist 
to be ranked higher than those that are not
by minimising a %the Bottom-Push
bipartite ranking loss.
Inspired by an equivalence between bipartite ranking and binary classification,
we show how one can efficiently approximate an optimal solution of the multitask learning objective by 
minimising a classification loss.
Empirical results on two real playlist datasets show the proposed approach has good performance 
for cold-start playlist recommendation.

\end{abstract}

\section{Introduction}
\label{sec:intro}
Online music streaming services (e.g., Spotify, Pandora, Apple Music) % Google Play Music, Amazon Music) 
are playing an increasingly important role in the digital music industry.
A key ingredient of these services is the ability to automatically recommend songs to help users explore large collections of music.
Such recommendation is often in the form of a \emph{playlist}, 
which involves a (small) set of songs.
We investigate the problem of recommending songs to form personalised playlists 
in \emph{cold-start} scenarios, %settings,
where there is no historical data for either users or songs.
Conventional recommender systems for books or movies~\citep{Sarwar:2001,Netflix}
typically learn a score function via matrix factorisation~\citep{Koren:2009},
and recommend the item that achieves the highest score.
This approach is not suited to %deal with 
cold-start settings
due to the lack of interaction data. %for either users or items.
Further, in playlist recommendation,
one has to recommend a subset of a large collection of songs instead of only one top ranked song.
Enumerating all possible such subsets is intractable;
additionally,
it is likely that more than one playlist is satisfactory, since
users generally maintain more than one playlist when using a music streaming service,
which leads to challenges in standard supervised learning.

We formulate playlist recommendation as a multitask learning problem.
Firstly, we study the setting of recommending %a set of songs to form %a new playlist for a user
personalised playlists for a user
by exploiting the (implicit) preference %of the given user 
from her existing playlists.
Since we do not have any contextual information about the new playlist, 
we call this setting \emph{cold playlists}.
We find that learning from a user's existing playlists %can significantly 
improves the accuracy of recommendation compared to 
suggesting popular songs from familiar artists.
We further consider the setting of \emph{cold users} (\ie new users),
where we recommend playlists for new users %by learning from playlists of existing users
given playlists from existing users.
We find it challenging to improve recommendations beyond simply ranking songs according to their popularity 
if we know nothing except the identifier of the new user, %about the user,
which is consistent with previous 
discoveries~\cite{mcfee2012million,bonnin2013evaluating,bonnin2015automated}.
However, improvement can still be achieved if we know a few simple attributes (\eg age, gender, country) % etc.)
of the new users.
Lastly, we investigate the setting of recommending newly released songs (\ie \emph{cold songs}) 
to extend users' existing playlists. 
We find that the set of songs in a playlist are particularly %especially 
helpful in guiding 
the selection of new songs to be added to the given playlist.

We propose a novel multitask learning method that %to %which can 
can deal with %handle
playlist recommendation in all three cold-start settings.
It optimises a bipartite ranking loss~\cite{Freund:2003,Agarwal:2005}
that encourages songs in a playlist to be ranked higher than those that are not.
This results in a convex optimisation problem with an enormous number of constraints.
Inspired by an equivalence between bipartite ranking and binary classification,
we efficiently approximate an optimal solution of the constrained objective 
by minimising an unconstrained classification loss.
We present experiments on two real playlist datasets, 
and demonstrate that our multitask learning approach improves over existing strong baselines 
for playlist recommendation in cold-start scenarios.

\section{Multitask learning for recommending playlists}
\label{sec:method}

\begin{figure*}[!t]
    \centering
    \begin{subfigure}[t]{2.3in}
        \centering
        \includegraphics[height=1.6in]{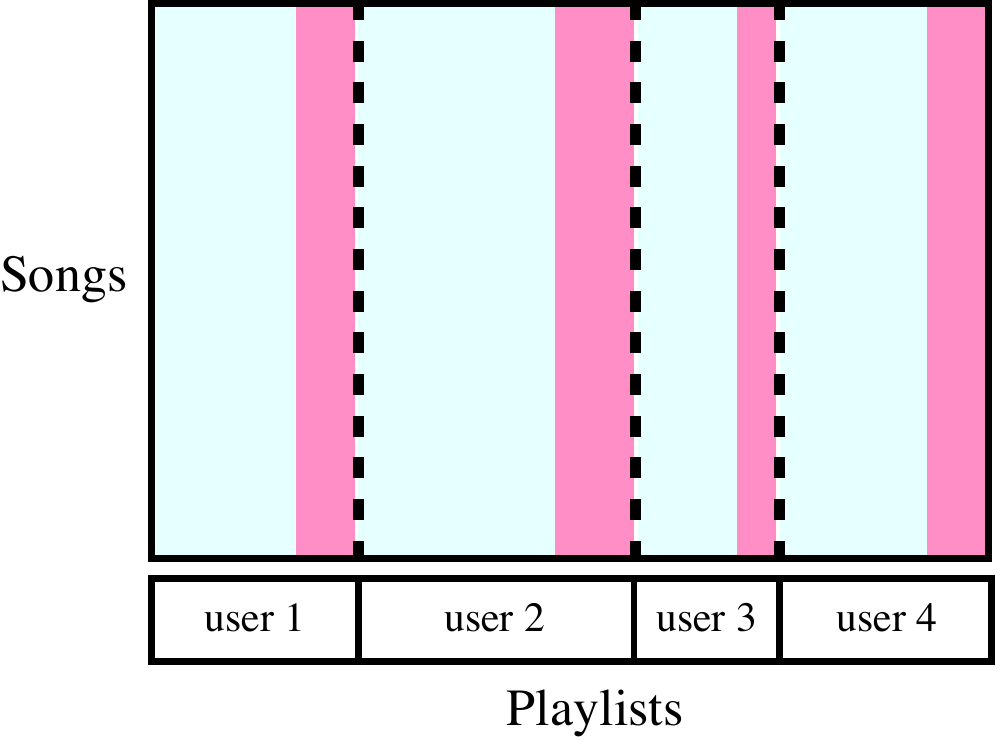}
        \caption{Cold Playlists}
    \end{subfigure}
    \begin{subfigure}[t]{2.3in}
        \centering
        \includegraphics[height=1.6in]{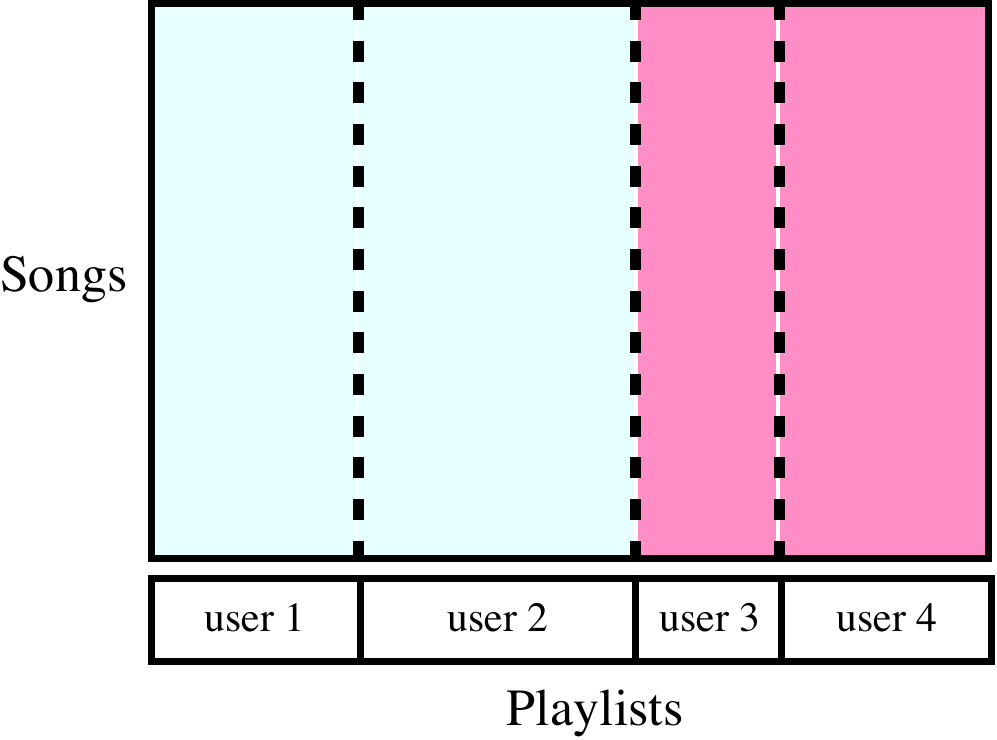}
        \caption{Cold Users}
    \end{subfigure}
    \begin{subfigure}[t]{2.3in}
        \centering
        \includegraphics[height=1.6in]{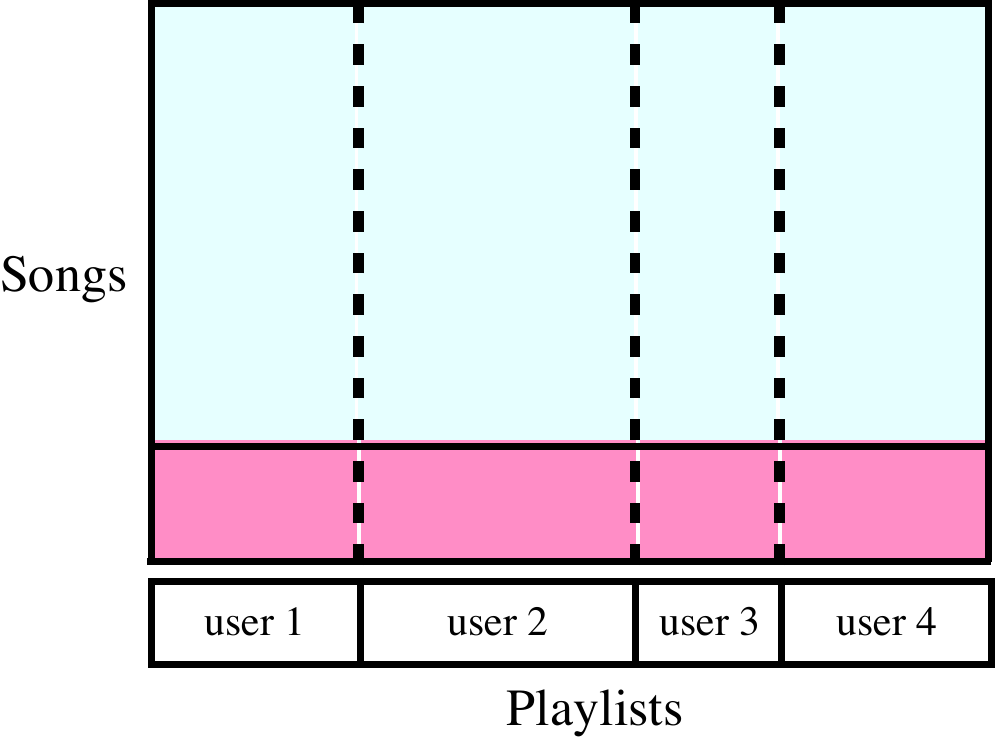}
        \caption{Cold Songs}
    \end{subfigure}
    %\caption{Three cold-start settings: training set (cyan), test set (magenta)}
	\caption{Three settings of cold-start playlist recommendation.
In each setting,
rows represent songs, and a column represents a playlist, 
which is a binary vector where an element denotes if the corresponding song is in the playlist.
%Playlists are first grouped by user, then split into training (Cyan) and test set (Magenta).
%(a) Cold Playlists: a portion of playlists from each user are held for test
%See text below for description.
Playlists are grouped by user.
{\color[rgb]{.4,1,1} \bf Light Cyan} represents playlists or songs in the training set, and
{\color[rgb]{1,0,.5} \bf dark Magenta} represents playlists or songs in the test set.
{\bf (a) Cold Playlists}: recommending personalised playlists (Magenta) for each user 
given users' existing playlists (Cyan); 
%by learning from her existing playlists (Cyan); 
{\bf (b) Cold Users}: recommending playlists for new users (Magenta) 
given playlists from existing users (Cyan);
%by learning playlists from existing users (Cyan); 
%by learning playlists from existing users (Cyan); 
{\bf (c) Cold Songs}: recommending newly released songs (Magenta) to extend users' existing playlists (Cyan).
}
\label{fig:setting}
\end{figure*}

We first define the three cold-start settings considered in this paper,
then introduce the multitask learning objective 
and show how the problem of cold-start playlist recommendation can be handled.
We discuss the challenge in optimising the multitask learning objective
via convex constrained optimisation
and show how one can efficiently approximate an optimal solution by minimising
an unconstrained objective.

\subsection{Cold playlists, cold users and cold songs}

Figure~\ref{fig:setting} illustrates the three cold-start settings for playlist recommendation 
that we study in this paper:
\begin{enumerate}[(a)]
\item Cold playlists, where we recommend songs to form new personalised playlists for each existing user;
\item Cold users, where we recommend songs to form new playlists for each new user;
\item Cold songs, where we recommend newly released songs to extend users' existing playlists.
\end{enumerate}

In the \emph{cold playlists} setting, a target user (\ie the one for whom we recommend playlists)
maintains a number of playlists that can be exploited by the learning algorithm.
In the \emph{cold users} setting, however, we may only know a few simple attributes of a new user
(\eg age, gender, country) or nothing except her user identifier. 
The learning algorithm can only make use of playlists from existing users.
Finally, in the \emph{cold songs} setting, the learning algorithm have access to 
content features (\eg artist, genre, audio data) 
of newly released songs as well as all playlists from existing users.

\subsection{Multitask learning objective}

Suppose we have
a dataset $\DCal$ with $N$ playlists from $U$ users, 
where songs in every playlist are from a music collection with $M$ songs.
Assume each user has at least one playlist, and each song in the collection 
appears in at least one playlist.
Let $P_u$ denote the (indices of) playlists from user $u \in \{1,\dots,U\}$.
We aim to learn a function $f(m, u, i)$ that measures the affinity between 
song $m \in \{1,\dots,M\}$ and playlist $i \in P_u$ from user $u$.
Suppose for song $m$, function $f$ has linear form,
\begin{equation}
\label{eq:scorefunc}
f(m, u, i) = \w_{u,i}^\top \x_m,
\end{equation}
where 
$\x_m \in \R^D$ represents the $D$ features of song $m$,
and $\w_{u,i} \in \R^D$ are the weights %(or representation) 
of playlist $i$ from user $u$.

Inspired by the decomposition of user weights and artist weights in~\cite{ben2017groove},
we decompose $\w_{u,i}$ 
into three components
\begin{equation}
\label{eq:decomp}
\w_{u, i} = \alphabm_u + \betabm_i + \bmu,
\end{equation}
where $\alphabm_u$ are weights for user $u$,
$\betabm_i$ are weights specific for playlist $i$, %(from user $u$),
and $\bmu$ are the weights shared by all users (and playlists).
This decomposition allows us to learn the user weights $\alphabm_u$ using all her
playlists, and exploit all training playlists when learning the shared weights $\bmu$.

Let $\upthetabm$ denote all parameters in $\left\{ \{\alphabm_u\}_{u=1}^U, \{\betabm_i\}_{i=1}^N, \bmu \right\}$.
The learning task is to minimise the empirical risk of affinity function $f$ on dataset $\DCal$ over $\upthetabm$,
\ie %we solve an optimisation problem
\begin{equation}
\label{eq:obj}
\min_\upthetabm \, \Omega(\upthetabm) + R_\upthetabm(f, \DCal),
\end{equation}
where $\Omega(\upthetabm)$ is a regularisation term and $R_\upthetabm(f, \DCal)$ denotes the empirical risk
of $f$ on $\DCal$.
We call the objective in problem~(\ref{eq:obj}) the {\it multitask learning objective},
since we jointly learn from multiple tasks where each one involves recommending a set of songs 
given a user or playlist.

We further assume that playlists from the \emph{same} user have \emph{similar} weights %representations,
and the shared weights $\bmu$ are sparse %that shared by all users are sparse 
(\ie users only share a small portion of their weights).
To impose these assumptions, we apply $\ell_1$ regularisation to encourage sparsity %sparse representations 
of the playlist weights $\betabm_i$ %from the same user
and the shared weights $\bmu$.
The regularisation term in our multitask learning objective is
\begin{equation*}
\Omega(\upthetabm) 
= \lambda_1 \sum_{u=1}^U \|\alphabm_u\|_2^2 
  + \lambda_2 \sum_{i=1}^N \|\betabm_i\|_1 
  + \lambda_3 \| \bmu \|_1,
\end{equation*}
where constants $\lambda_1, \lambda_2, \lambda_3 \in \R_+$,
and the $\ell_2$ regularisation term is to penalise large values in user weights. %representations.
We specify the empirical risk $R_{\uptheta}(f, \DCal)$ later.

\subsection{Cold-start playlist recommendation}

Once parameters $\uptheta$ have been learned, 
we make a recommendation by first scoring each song according to available information (\eg an existing user or playlist),
then form or extend a playlist by either taking the top-$K$ scored songs or sampling songs with probabilities proportional to their scores.
Specifically, %to recommend a new playlist for an existing user (\ie the \emph{cold playlists} setting),
in the \emph{cold playlists} setting where the target user $u$ is known,
we score song $m$ as 
\begin{equation}
\label{eq:cp}
r_m^{(a)} = (\alphabm_u + \bmu)^\top \x_m.
\end{equation}

Further, in the \emph{cold users} setting where %a few 
simple attributes of the new user are available,
we approximate the weights of the new user using the average weights of similar existing users
(\eg in terms of cosine similarity of user attributes)
and score song $m$ as
\vspace{-.8em}
\begin{equation}
\label{eq:cu}
r_m^{(b)} = \left( \frac{1}{|\UCal|} \sum_{u \in \UCal} \alphabm_u + \bmu \right)^\top \x_m,
\end{equation}
where $\UCal$ is the set of (\eg 10) existing users that are most similar to the new user.
On the other hand, if we know nothing about the new user except her identifier,
we can simply score song $m$ using the shared weights, \ie
%\vspace{-.3em}
\begin{equation}
\label{eq:cu2}
r_m^{(b)} = \bmu^\top \x_m.
\end{equation}

Lastly, %to extend an existing playlist $i$ from user $u$ (\ie the \emph{cold songs} setting),
in the \emph{cold songs} setting where we are given a specific playlist $i$ from user $u$,
we therefore can score song $m$ using both user weights and playlist weights, \ie
%\vspace{-.3em}
\begin{equation}
\label{eq:cs}
r_m^{(c)} = (\alphabm_u + \betabm_i + \bmu)^\top \x_m.
\end{equation}

We now specify the empirical risk $R_{\uptheta}(f, D)$ %in problem~(\ref{eq:obj})
and develop methods to optimise the multitask learning objective.

\subsection{Constrained optimisation with ranking loss}

We aim to rank songs that are likely in a playlist above those that are unlikely when making a recommendation.
To achieve this, we optimise the multitask learning objective by minimising a bipartite ranking loss.
In particular, we minimise the number of songs not in a training playlist but ranked above the lowest 
ranked song in it.\footnote{This is known as the Bottom-Push~\cite{rudin2009p} in the bipartite ranking literature.}
The loss of the affinity function $f$ for playlist $i$ from user $u$ is defined as
\begin{equation*}
%\label{eq:loss}
\resizebox{\columnwidth}{!}{$
\Delta_f(u, i) 
= \displaystyle \frac{1}{M_-^i} \sum_{m': y_{m'}^i = 0} \llb \min_{m: y_m^i = 1} f(m, u, i) \le f(m', u, i) \rrb,
$}
\end{equation*}
where $M_-^i$ is the number of songs not in playlist $i$,
binary variable $y_m^i$ denotes whether song $m$ appears in playlist $i$,
and $\llb \cdot \rrb$ is the indicator function that represents the 0/1 loss.

The empirical risk when employing the bipartite ranking loss %~(\ref{eq:loss}) is
$\Delta_f(u, i)$ is
\vspace{-1em}
\begin{equation}
\label{eq:risk_rank}
R_{\upthetabm}^{\textsc{rank}}(f, \DCal) = \frac{1}{N} \sum_{u=1}^U \sum_{i \in P_u} \Delta_f(u, i).
\end{equation}

There are two challenges %to %optimise the above objective,
when optimising the multitask learning objective in problem~(\ref{eq:obj}) 
with the empirical risk $R_{\uptheta}^{\textsc{rank}}$,
namely, the non-differentiable 0/1 loss and the \emph{min} function in %$R_{\upthetabm}^{\textsc{rank}}$.
$\Delta_f(u, i)$.
To address these challenges, we first upper-bound %replace 
the 0/1 loss with one of its convex surrogates, 
\eg the exponential loss $\llb z \le 0 \rrb \le e^{-z}$,
\begin{equation*}
\resizebox{\columnwidth}{!}{$
\displaystyle 
\Delta_f(u, i) \le 
\frac{1}{M_-^i} \sum_{m': y_{m'}^i = 0} \!\! \exp \left(f(m', u, i) - \!\! \min_{m: y_m^i = 1} f(m, u, i) \right).
$}
\end{equation*}

One approach to deal with the \emph{min} function in $\Delta_f(u, i)$ is introducing 
slack variables $\xi_i$ to lower-bound the scores of songs in playlist $i$ 
and 
transform problem (\ref{eq:obj}) with empirical risk $R_{\uptheta}^{\textsc{rank}}$ into a convex constrained optimisation problem 
\begin{equation*}
\resizebox{\columnwidth}{!}{$\displaystyle
\begin{aligned}
\min_{\upthetabm} \ \, & 
\Omega(\upthetabm) 
+ \frac{1}{N} \sum_{u=1}^U \sum_{i \in P_u} \frac{1}{M_-^i} \sum_{m': y_{m'}^i = 0} \!\! \exp \left( f(m', u, i) - \xi_i \right) \\
s.t. \ \, & 
\xi_i \le f(m, u, i), \\
& u \in \{1,\dots,U\}, \ i \in P_u, \ m \in \{1,\dots,M\} \, \text{and} \ y_m^i = 1.
\end{aligned}
$}
\end{equation*}

Note that the number of constraints in the above optimisation problem is
$
\sum_{u=1}^U \sum_{i \in P_u} \sum_{m=1}^M \, \llb y_m^i \!= \!1 \rrb,
$
\ie the accumulated playcount of all songs,
which is of order {\small $O(\widebar{L} N)$} asymptotically, where {\small $\widebar{L}$} 
is the average number of songs in playlists (typically less than $100$). 
However,
the total number of playlists {\small $N$} can be enormous 
in production systems 
(\eg Spotify hosts more than $2$ billion playlists\footnote{https://newsroom.spotify.com/companyinfo}),
which imposes a significant challenge in optimisation. 
This issue could be alleviated by applying the cutting-plane method~\cite{avriel2003nonlinear} or the sub-gradient method.
Unfortunately, we find both methods converge extremely slowly for this problem in practice.
In particular, the cutting plane method is required to solve 
a constrained optimisation problem with at least {\small $N$} constraints in each iteration, 
which remains challenging.

\subsection{Unconstrained optimisation with classification loss}

An alternative approach to deal with the \emph{min} function in $\Delta_f(u, i)$ 
is approximating it using the well known Log-sum-exp function~\cite[p. 72]{boyd2004convex},
\begin{equation*}
  \displaystyle \min_j z_j 
= \displaystyle -\max_j (-z_j) 
= \displaystyle -\lim_{p \to +\infty} \frac{1}{p} \log \sum_j \exp(-p z_j),
\end{equation*}
which allows us to approximate the empirical risk $R_{\upthetabm}^{\textsc{rank}}$ (with the exponential surrogate)
by $\widetilde R_{\upthetabm}^{\textsc{rank}}$ defined as
\begin{equation*}
\resizebox{\columnwidth}{!}{$
\begin{aligned}
\widetilde R_{\upthetabm}^{\textsc{rank}}(f, \DCal)
= \frac{1}{N} \sum_{u=1}^U \sum_{i \in P_u} \frac{1}{M_-^i} \left[
  \sum_{m: y_m^i = 1} \Big[ \delta_f(m, u, i) \Big]^p \right]^\frac{1}{p} \!\! ,
\end{aligned}
$}
\end{equation*}
where hyper-parameter $p \in \R_+$ and
$$
\delta_f(m, u, i) = \!\! \sum_{m': y_{m'}^i = 0} \!\! \exp(-(f(m, u, i) - f(m', u, i))).
$$

We further observe that $\widetilde R_{\upthetabm}^{\textsc{rank}}$ can be transformed into the standard P-Norm Push loss~\cite{rudin2009p} 
by simply swapping the positives {\small $\{m: y_m^i = 1\}$} and negatives {\small $\{m': y_{m'}^i = 0\}$}.
Inspired by the connections between bipartite ranking and binary classification~\cite{menon2016bipartite},
we swap the positives and negatives in the P-Classification loss~\cite{ertekin2011equivalence} while taking care of signs.
This results in an empirical risk with a classification loss:
\begin{equation}
\label{eq:clfrisk}
\resizebox{\columnwidth}{!}{$\displaystyle
\begin{aligned}
R_{\upthetabm}^{\textsc{mtc}}(f, \DCal)
= \frac{1}{N} \sum_{u=1}^U 
& \sum_{i \in P_u} \Bigg(
  \frac{1}{p M_+^i} \sum_{m: y_m^i = 1} \!\! \exp(-p f(m, u, i)) \\
& \hspace{2em}  + \frac{1}{M_-^i} \sum_{m': y_{m'}^i = 0} \!\! \exp(f(m', u, i)) \Bigg),
\end{aligned}
$}
\end{equation}
where $M_+^i$ is the number of songs in playlist $i$.

\begin{lemma}
\label{lm:rank2clf}
Let $\upthetabm^* \in \argmin_{\upthetabm} R_{\upthetabm}^{\textsc{mtc}}$ (assuming minimisers exist),
then $\upthetabm^* \in \argmin_{\upthetabm} \widetilde R_{\upthetabm}^{\textsc{rank}}$.
\end{lemma}

\begin{proof}
See Appendix for a complete proof.
Alternatively, 
we can use the proof of the equivalence between P-Norm Push loss and P-Classification loss~\cite{ertekin2011equivalence}
if we swap the positives and negatives.
\end{proof}

By Lemma~\ref{lm:rank2clf}, 
we can optimise the parameters of the multitask learning objective by 
solving a (convex) unconstrained optimisation problem:\footnote{We choose not to directly optimise 
the empirical risk $\widetilde R_{\uptheta}^{\textsc{rank}}$, which involves the P-Norm Push, 
since classification loss can be optimised more efficiently in general~\cite{ertekin2011equivalence}.}
\begin{equation}
\label{eq:expobj_clf}
\min_\upthetabm \ \Omega(\upthetabm) + R_{\upthetabm}^{\textsc{mtc}}(f, \DCal).
\end{equation}

Problem~(\ref{eq:expobj_clf}) can be efficiently optimised using the
Orthant-Wise Limited-memory Quasi-Newton (OWL-QN) algorithm~\cite{andrew2007scalable},
an L-BFGS variant that can address $\ell_1$ regularisation effectively.

We refer to the approach that solves problem~(\ref{eq:expobj_clf}) as \emph{Multitask Classification} (MTC). 
As a remark, optimal solutions of problem (\ref{eq:expobj_clf}) are not necessarily the optimal solutions 
of problem $\min_\upthetabm \ \Omega(\upthetabm) + \widetilde R_{\upthetabm}^{\textsc{rank}}$ due to regularisation. 
However, when parameters $\upthetabm$ are small (which is generally the case when using regularisation),
optimal solutions of the two objectives can nonetheless approximate each other well.

\section{Related work}

We summarise recent work most related to playlist recommendation and music recommendation in cold-start 
scenarios,
as well as work on the connection between bipartite ranking and binary classification.

There is a rich collection of recent literature on playlist recommendation,
which can be summarised into two typical settings: 
playlist generation and next song recommendation. % and playlist continuation.
Playlist generation is to produce a complete playlist given some 
seed. %~\cite{platt2002learning,mcfee2011natural,mcfee2012hypergraph,chen2012playlist,ben2017groove},
For example, the AutoDJ system~\cite{platt2002learning} generates playlists given one or more seed songs;
Groove Radio can produce a personalised playlist for the specified user given a seed artist~\cite{ben2017groove};
or a seed location in hidden space (where all songs are embedded)
can be specified in order to generate a complete playlist~\cite{chen2012playlist}.
There are also works that focus on evaluating the learned playlist model,
without concretely generating playlists~\cite{mcfee2011natural,mcfee2012hypergraph}.
See this recent survey~\cite{bonnin2015automated} for more details.

Next song recommendation %~\cite{hariri2012context,bonnin2013evaluating,jannach2015beyond}
predicts
the next song a user might play after observing some context.
For example, the most recent sequence of songs with which a user has interacted was used to
infer the contextual information,
which was then employed %has been employed %was further adopted 
to rank the next possible song
via %with regards to 
a topic-based sequential model %patterns 
learned from users' playlists~\cite{hariri2012context}.
Context can also be the artists %appeared 
in a user's listening history,
which has been %were 
employed to score the next song together with frequency of artist collocations
as well as song popularity~\cite{mcfee2012million,bonnin2013evaluating}.
It is straightforward to produce a complete playlist using next song recommendation techniques,
\ie by picking the next song sequentially~\cite{bonnin2013evaluating,ben2017groove}.

In the collaborative filtering literature,
the cold-start setting has primarily been addressed through
suitable regularisation of matrix factorisation parameters
based on exogenous user- or item-features~\cite{Ma:2008,Agarwal:2009,Cao:2010}.
Content-based approaches~\cite[chap. 4]{aggarwal2016recommender}
can handle the recommendation of new songs,
typically by making use of content features of songs extracted either automatically~\cite{seyerlehner2010automatic,eghbal2015vectors}
or manually by musical experts~\cite{john2006pandora}.
Further, content features can also be combined with other approaches, such as those based on 
collaborative filtering~\cite{yoshii2006hybrid,donaldson2007hybrid,shao2009music},
which is known as the hybrid recommendation approach~\cite{burke2002hybrid,aggarwal2016recommender}.

Another popular approach for cold-start recommendation involves explicitly mapping 
user- or item- content features to latent embeddings~\cite{Gantner:2010}.
This approach can be adopted to recommend new songs, 
\eg by learning a convolutional neural network to map audio features of new songs to 
the corresponding latent embeddings~\cite{van2013deep},
which were then %are %further 
used to score songs together with the latent embeddings of playlists (learned by MF). % matrix factorisation).
The problem of recommending music for new users can also be tackled using a similar approach, \eg
by learning a mapping from user attributes to user embeddings.

A slightly different approach to deal with music recommendation for new users is learning hierarchical 
representations for genre, sub-genre and artist.
By adopting an additive form with user and artist weights, it can fall back to using only artist weights
when recommending music to new users; if the artist weights are not available (\eg a new artist), this approach 
further falls back to using the weights of sub-genre or genre~\cite{ben2017groove}.
However, the requirement of seed information (\eg artist, genre or a seed song) restricts its direct applicability to
the \emph{cold playlists} and \emph{cold users} settings. %where no seed information is available.
Further, encoding song usage information as features makes it unsuitable 
for recommending new songs directly.

It is well known
that bipartite ranking and binary classification are
closely related~\cite{ertekin2011equivalence,menon2016bipartite}.
In particular, \citet{ertekin2011equivalence} have shown that the P-Norm Push~\cite{rudin2009p}
is equivalent to the P-Classification when 
the exponential surrogate of 0/1 loss is employed.
Further, the P-Norm Push is an approximation of the Infinite-Push~\cite{agarwal2011infinite},
or equivalently, the Top-Push~\cite{li2014top}, which focuses on the highest ranked negative example instead of
the lowest ranked positive example in the Bottom-Push adopted in this work.
Compared to the Bayesian Personalised Ranking (BPR) approach~\cite{rendle2009bpr,mcfee2012million} that requires all
positive items to be ranked higher than those unobserved ones, 
the adopted approach only penalises unobserved items that ranked higher than the lowest ranked positive item,
which can be optimised more efficiently 
when only the top ranked items are of interest~\cite{rudin2009p,li2014top}.

\clearpage
\newpage

\section{Experiments}
\label{sec:experiment}

We present empirical evaluations for cold-start playlist recommendation on two real playlist datasets,
and compare the proposed multitask learning method with a number of well known baseline approaches.

\subsection{Dataset}
We evaluate on 
two publicly available playlist datasets: the 30Music~\cite{30music2015} and 
the AotM-2011~\cite{mcfee2012hypergraph} dataset.
The Million Song Dataset (MSD)~\cite{msd2011} serves as an underlying dataset where songs in all playlists 
are intersected; additionally, song and artist information in the MSD are used to compute song features.

\noindent
{\bf 30Music Dataset} is a collection of listening events and user-generated
playlists retrieved from Last.fm.\footnote{https://www.last.fm}
We first intersect the playlists data with songs in the MSD, 
then filter out playlists with less than 5 songs.
This results in about 17K playlists over 45K songs from 8K users.

\noindent
{\bf AotM-2011 Dataset} is a collection of playlists shared by Art of the Mix\footnote{http://www.artofthemix.org} 
users during the period from 1998 to 2011. Songs in playlists have been matched to those in the MSD.
It contains 
roughly 84K playlists over 114K songs from 14K users
after filtering out playlists with less than 5 songs.

Table~\ref{tab:stats_pldata} summarises the two 
playlist datasets used in this work.
See Appendix for more details.

\subsection{Features}

Song metadata, audio data, genre and artist information, as well as song popularity
(\ie the accumulated playcount of the song in the training set)
and artist popularity 
(\ie the accumulated playcount of all songs from the artist in the training set)
are encoded as features.
The metadata of songs (\eg duration, year of release) and pre-computed audio features (\eg loudness, mode, tempo) are from the MSD.
We use genre data from the Top-MAGD genre dataset~\cite{schindler2012facilitating}
and tagtraum genre annotations for the MSD~\cite{schreiber2015improving} via one-hot encoding.
If the genre of a song is unknown, 
we apply mean imputation using genre counts of songs in the training set.
To encode artist information as features,
we create a sequence of artist identifiers for each playlist in the training set, and train
a word2vec\footnote{https://github.com/dav/word2vec} model that learns embeddings of artists.
We assume no popularity information is available for newly released songs,
and therefore song popularity is not a feature in the \emph{cold songs} setting.
Finally, we add a constant feature (with value $1.0$) for each song to account for bias.

\begin{table}[!t]
\centering
\caption{Statistics of music playlist datasets}
\label{tab:stats_pldata}
\resizebox{.8\linewidth}{!}{
\begin{tabular}{lrr}
\toprule
               & 30Music & AotM-2011 \\
\midrule
Playlists      & 17,457  & 84,710    \\
Users          & 8,070   & 14,182    \\
Avg. Playlists per User & 2.2     & 6.0       \\
\midrule
Songs          & 45,468  & 114,428   \\
Avg. Songs per Playlist & 16.3    & 10.1      \\
\midrule
Artists        & 9,981   & 15,698    \\
Avg. Songs per Artist   & 28.6    & 53.8      \\
\bottomrule
\end{tabular}
}
\end{table}

\subsection{Experimental setup}

We first split the two playlist datasets into training and test sets, %for each of the three cold-start settings,
then evaluate the test set performance of the proposed method,
and compare it against
several baseline approaches
in each of the three cold-start settings.

\subsubsection{Dataset split}
In the \emph{cold playlists} setting,
we hold a portion of the playlists from about 20\% of users in both datasets for testing, 
and all other playlists are used for training.
The test set is formed by sampling playlists where each song has been included in 
at least five playlists among the whole dataset.
We also make sure each song in the test set appears in the training set,
and all users in the test set have a few %a number of 
playlists in the training set.
In the \emph{cold users} setting,
we sample 30\% of users and hold all of their playlists in both datasets.
Similarly, we require songs in the test set to appear in the training set,
and a user will thus not be used for testing %included in the test set 
if holding all of her playlists breaks this requirement.
To evaluate 
in the \emph{cold songs} setting,
we hold 5K of the latest released songs in the 30Music dataset,
and 10K of the latest released songs in the AotM-2011 dataset where more songs are available.
We remove playlists where all songs have been held for testing. %are removed from both the training and test set.

See Appendix for the statistics of these dataset splits.

\begin{table*}[!t]
\caption{AUC for playlist recommendation in three cold-start settings. \emph{Higher} values indicate better performance.}
\label{tab:auc}
\resizebox{\textwidth}{!}{
\begin{tabular}{lcccclcccclcccc}
\toprule 
\multicolumn{3}{c}{Cold Playlists} &&& \multicolumn{3}{c}{Cold Users} &&& \multicolumn{3}{c}{Cold Songs} \\ 
\cmidrule{1-3} \cmidrule{6-8} \cmidrule{11-13}
Method &  30Music &  AotM-2011 &&& 
Method &  30Music &  AotM-2011 &&& 
Method &  30Music &  AotM-2011 \\
\cmidrule{1-3} \cmidrule{6-8} \cmidrule{11-13}
%\midrule
PopRank &  94.0    &  93.8      &&& 
PopRank &  88.3    & {\bf 91.8} &&& 
PopRank &  70.9    &  76.5      \\
CAGH    &  94.8    &  94.2      &&& 
CAGH    &  86.3    &  88.1      &&& 
CAGH    &  68.0    &  77.4      \\
SAGH    &  64.5    &  79.8      &&&
SAGH    &  54.5    &  53.7      &&&
SAGH    &  51.5    &  53.6      \\
WMF     &  79.5    &  85.4      &&&
WMF+kNN &  84.9    &  N/A       &&&
MF+MLP  &  81.4    &  80.8      \\
MTC     & {\bf 95.9} & {\bf 95.4}  &&&
MTC     & {\bf 88.8} & {\bf 91.8}  &&&
MTC     & {\bf 86.6} & {\bf 84.3}  \\
\bottomrule
\end{tabular}
}
\end{table*}

\subsubsection{Baselines}
We compare the performance of our proposed method (\ie MTC) % (referred as {\it Multitask Classification}) 
with the following %a number of 
baseline approaches in each of the three cold-start settings:
\begin{itemize}
\item The {\it Popularity Ranking} (PopRank) method scores a song using only its popularity in the training set.
      In the \emph{cold songs} setting where song popularity is not available, 
      a song is scored by the popularity of the corresponding artist.
\item The {\it Same Artists - Greatest Hits} (SAGH)~\cite{mcfee2012million} method scores a song
      by its popularity if the artist of the song appears in the given user's playlists (in the training set);
      otherwise the song is scored zero.
      In the {\it cold songs} setting, this method only considers songs from artists that appear in the given playlist,
      and scores a song using the popularity of the corresponding artist.
\item The {\it Collocated Artists - Greatest Hits} (CAGH)~\cite{bonnin2013evaluating} method is a variant of SAGH.
      It scores a song using its popularity, but weighted by the frequency of the collocation between 
      the artist of the song and artists that appear in the given user's playlists (in the training set).
      In the \emph{cold users} setting, we use the 10 most popular artists instead of artists 
      in the user's listening history, and the \emph{cold songs} setting is addressed in the same way as in SAGH.
\item A variant of Matrix Factorisation (MF), which first learns the latent factors of songs, playlists
      or users through MF, then scores each song by the dot product of the corresponding latent factors.
      Recommendations are made as per the proposed method.
      In the \emph{cold playlists} setting, we factorise the song-user playcount matrix using the 
      weighted matrix factorisation (WMF) algorithm~\cite{hu2008collaborative}, which learns the 
      latent factors of songs and users.
      In the \emph{cold users} setting, we first learn the latent factors of songs and users using WMF,
      then approximate the latent factors of a new user by the average latent factors of the $k$ (\eg 100)
      nearest neighbours %\footnote{We choose $k=100$}
      (in terms of cosine similarity of user attributes, \eg age, gender and country) in the training set.
      We call this method WMF+kNN.
      \footnote{This method does not apply to the AotM-2011 dataset in the cold users setting,
      since such user attributes (\eg age, gender and country) are not available in the dataset.}
      In the \emph{cold songs} setting, we factorise the song-playlist matrix to learn the latent factors of 
      songs and playlists, which are then used to train a neural network to map song content features 
      to the corresponding latent factors~\cite{Gantner:2010,van2013deep}.
	  We can then obtain the latent factors of a new song as long as its content features are available.
      We call this method MF+MLP. 
\end{itemize}

\subsubsection{Evaluation}

We evaluate %the performance of 
all approaches using two accuracy metrics that have been adopted 
in playlist recommendation tasks:
\emph{HitRate@K}~\cite{hariri2012context} and \emph{Area under the ROC curve} (AUC)~\cite{manning2008introIR}.
We further adopt two beyond-accuracy metrics: %\ie
\emph{Novelty}~\cite{zhang2012auralist,schedl2017} and \emph{Spread}~\cite{kluver2014evaluating},
which are specifically tailored to recommender systems.

\emph{HitRate@K} (\ie Recall@K) 
is the number of correctly recommended songs amongst the top-$K$ recommendations over
the number of songs in the %ground truth playlist. %\footnote{This metric is also known as Recall@K~\cite{schedl2017}.}.
observed playlist.
It has been widely employed to evaluate %several 
playlist generation and next song recommendation
methods~\cite{hariri2012context,bonnin2013evaluating,bonnin2015automated,jannach2015beyond}.
AUC has been primarily used to measure the performance of classifiers.
It has been applied to evaluate %performance of 
playlist generation methods when the task
has been cast as a sequence of classification problems~\cite{ben2017groove}.

It is believed that
useful recommendations need to include previously unknown items~\cite{herlocker2004evaluating,zhang2012auralist}.
This ability can be measured by \emph{Novelty},
which is based on the assumption that, intuitively, the more popular a song is, 
the more likely a user is to be familiar with it, and therefore the less likely to be novel.
\emph{Spread}, however, is used to measure the ability of an algorithm to spread its attention across all possible songs.
It is defined as the entropy of the distribution of all songs.
See Appendix for more details of these beyond-accuracy metrics.

\subsection{Results and discussion}

\begin{figure}[!t]
    \centering
    \includegraphics[width=\columnwidth]{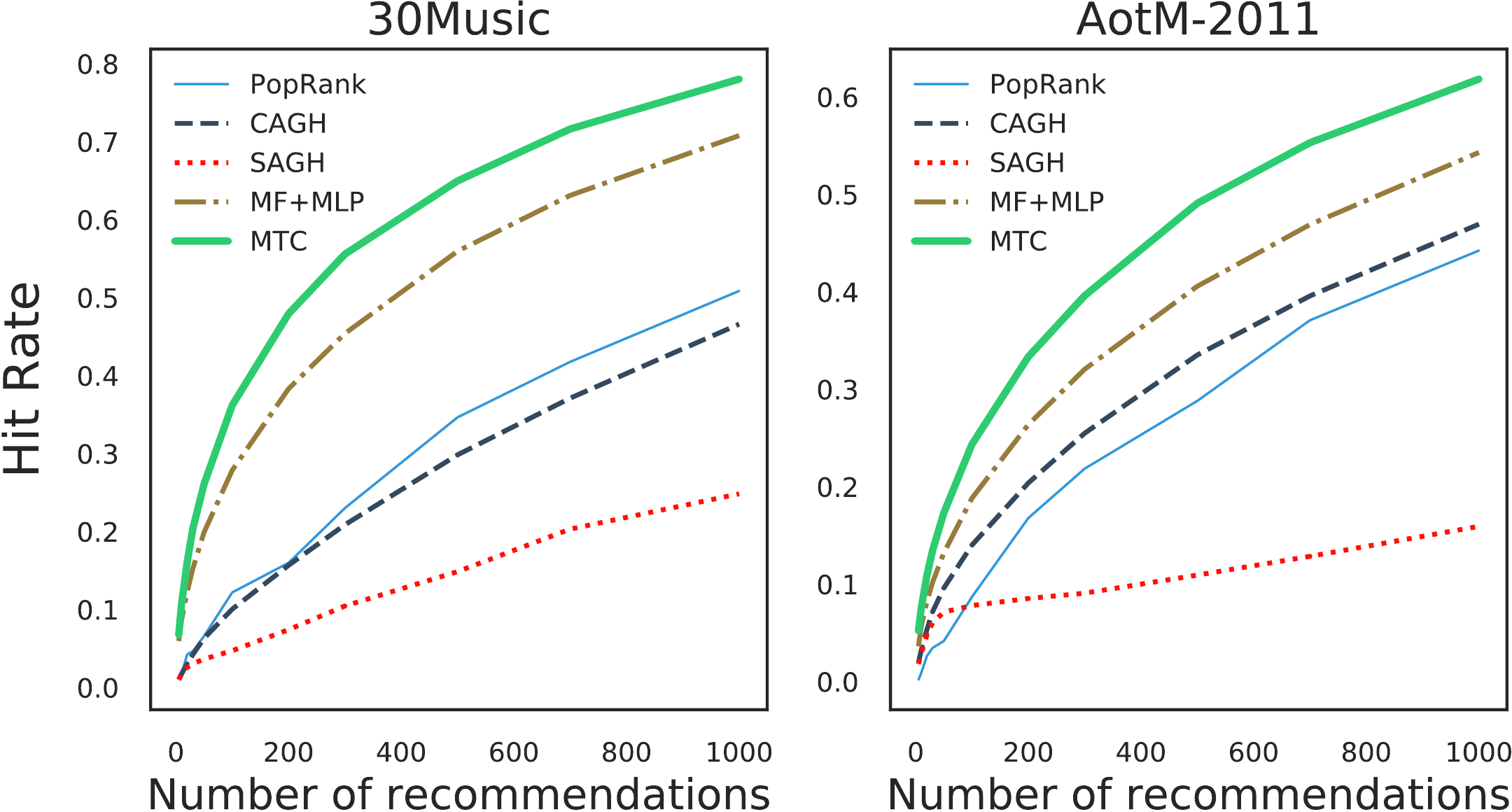}
    \caption{Hit rate of recommendation in the \emph{cold songs} setting.
\emph{Higher} values indicate better performance.}
    \label{fig:hr1}
\end{figure}

\begin{table*}[!t]
\caption{Spread for playlist recommendation in three cold-start settings. \emph{Moderate} values are preferable.}
\label{tab:spread}
\resizebox{\textwidth}{!}{
\begin{tabular}{lcccclcccclcccc}
\toprule 
\multicolumn{3}{c}{Cold Playlists} &&& \multicolumn{3}{c}{Cold Users} &&& \multicolumn{3}{c}{Cold Songs} \\ 
\cmidrule{1-3} \cmidrule{6-8} \cmidrule{11-13}
Method & 30Music & AotM-2011 &&& Method & 30Music & AotM-2011 &&& Method & 30Music & AotM-2011 \\ 
\cmidrule{1-3} \cmidrule{6-8} \cmidrule{11-13}
%\midrule
PopRank & \: 9.8 &   10.5  &&&  
PopRank & \: 9.8 &   10.5  &&&  
PopRank &    7.4 &    7.8  \\
CAGH    & \: 5.8 & \: 2.3  &&&  
CAGH    & \: 4.2 & \: 5.3  &&&  
CAGH    &    4.3 &    4.6  \\
SAGH    &   10.3 &   10.4  &&&  
SAGH    &   10.0 &   10.7  &&&  
SAGH    &    6.5 &    5.9  \\
WMF     &   10.7 &   11.6  &&&  
WMF+kNN &   10.7 & N/A     &&&
MF+MLP  &    8.5 &    9.2  \\
MTC     & \: 9.4 &   10.4  &&&  
MTC     & \: 9.9 &   11.4  &&&  
MTC     &    7.9 &    8.3  \\
\bottomrule
\end{tabular}
}
\end{table*}

\subsubsection{Accuracy}

Table~\ref{tab:auc} shows the performance of all methods in terms of AUC.
We can see that PopRank %which simply ranking songs based on (song or artist) popularity 
achieves good performance in all three cold-start settings.
This is in line with results reported in~\cite{bonnin2013evaluating,bonnin2015automated}.
Artist information, particularly the frequency of artist collocations that is exploited in CAGH, 
improves recommendation in the cold playlists and cold songs settings.
Further, PopRank is one of the best performing methods in the cold users setting,
which is consistent with previous discoveries~\cite{mcfee2012million,bonnin2013evaluating,bonnin2015automated}.
The reason is believed to be the long-tailed distribution of songs in
playlists~\cite{cremonesi2010performance,bonnin2013evaluating}.
The MF variant does not perform well in the cold playlists setting,
but it performs reasonably well in the cold users setting when attributes of new users are available
(\eg in the 30Music dataset),
and it works particularly well in the cold songs setting where both song metadata and audio features are available 
for new songs.

Lastly, 
MTC is the (tied) best performing method in all three cold-start settings on both datasets.
Interestingly, it achieves the same performance as PopRank in the cold users setting on the AotM-2011 dataset,
which suggests that MTC might degenerate to simply ranking songs according to popularity when making recommendations
for new users; however, when attributes of new users are available, %(\eg in the 30Music dataset), 
it can improve 
by exploiting information learned from existing users.

Figure~\ref{fig:hr1} shows the hit rate of all methods in the cold songs setting
when the number of recommended new songs varies from 5 to 1000.
As expected, the performance of all methods improves when the number of recommendations increases.
Further, we observe that learning based approaches (\ie MTC and MF+MLP) always perform better than 
other baselines that use only artist information.
works surprisingly well;
it even outperforms CAGH which exploits artist collocations on the 30Music dataset.
The fact that CAGH always performs better than SAGH confirms that artist collocation is helpful
for music recommendation.
Lastly, MTC outperforms all other methods by a big margin on both datasets,
which demonstrates the effectiveness of the proposed approach for recommending new songs.

We also observe that MTC improves over baselines in the cold playlists
and cold users settings (when simple attributes of new users are available),
although the margin is not as big as that in the cold songs setting.
See Appendix for details.

\subsubsection{Beyond accuracy}

Note that, unlike AUC and hit rate,
where higher values indicate better performance,
\emph{moderate} values of Spread and Novelty are usually preferable~\cite{kluver2014evaluating,schedl2017}.

Table~\ref{tab:spread} shows the performance of all recommendation approaches in terms of \emph{Spread}.
In the cold songs setting, CAGH and SAGH focus on songs from artists in users' listening history and similar artists, 
which explains the relative low \emph{Spread}.
However, in the cold playlists and cold users settings, 
SAGH improves its attention spreading due to the set of songs it focuses on is significantly bigger 
(\ie songs from all artists in users' previous playlists and songs from the 10 most popular artists, respectively).
Surprisingly, CAGH remains focusing on a relatively small set of songs in both settings.
Lastly, in all three cold-start settings, the MF variants have the highest \emph{Spread},
while both PopRank and MTC have (similar) moderate \emph{Spread},
which is considered better. %to have better performance.

Figure~\ref{fig:nov3} shows the \emph{Novelty} of all methods in the cold playlists setting.
We can see that PopRank has the lowest \emph{Novelty},
which is not surprising given the definition of \emph{Novelty} (see Appendix).
Both SAGH and CAGH start with low \emph{Novelty} and grow when the number of recommended songs increases,
but the \emph{Novelty} of CAGH saturates much earlier than that of SAGH.
The reason could be that,
when the number of recommendations is larger than the total number of songs from artists in a user's previous playlists,
SAGH will simply recommend songs randomly (which are likely to be novel)
while CAGH will recommend songs from artists that are similar to those in the user's previous playlists
(which could be comparably less novel).
Further, MTC achieves lower \emph{Novelty} than WMF and CAGH, 
which indicates that MTC tends to recommend popular songs to form new playlists. %for existing users.
To conclude, 
MTC and CAGH have moderate \emph{Novelty} on both datasets,
and therefore perform better than other approaches.

The proposed approach also achieves moderate \emph{Novelty} in the cold songs setting.
However, in the cold users setting, the MF variant and CAGH have moderate \emph{Novelty},
which are therefore preferred. See Appendix for details.

\begin{figure}[!t]
    \centering
    \includegraphics[width=\columnwidth]{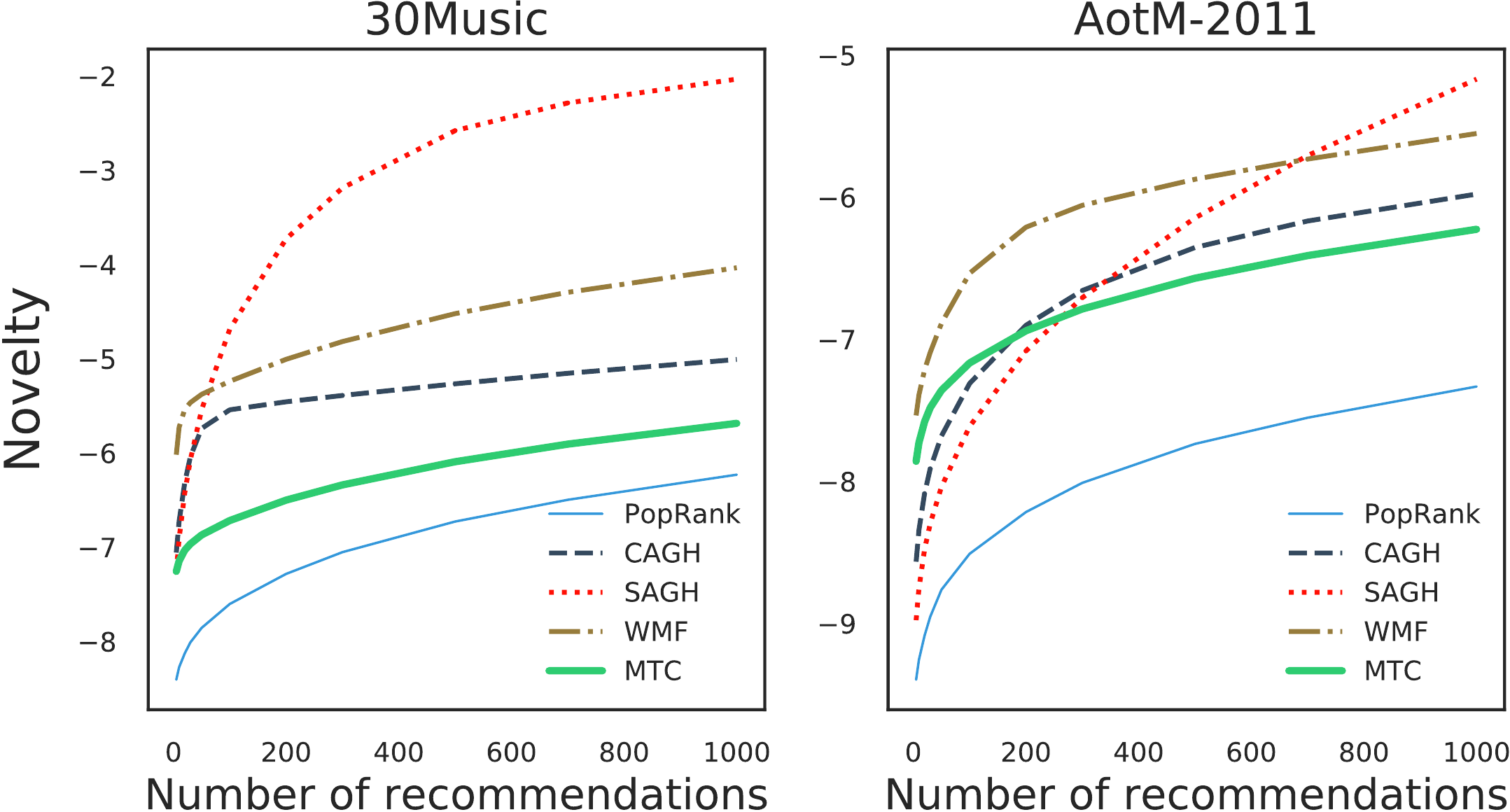}
    \caption{Novelty of recommendation in the \emph{cold playlists} setting.
\emph{Moderate} values are preferable.}
    \label{fig:nov3}
\end{figure}

\section{Conclusion and future work}

We study the problem of recommending playlists to users in three cold-start settings:
cold playlists, cold users and cold songs.
We propose a multitask learning method that learns user- and playlist-specific weights 
as well as shared weights from user-curated playlists,
which allows us to form new personalised playlists for an existing user, %in the cold playlists setting,
produce playlists for a new user, %in the cold users setting, 
and extend users' playlists with newly released songs. %in the cold songs setting.
We optimise the parameters (\ie weights) %of the multitask learning method 
by minimising a bipartite ranking loss
that encourages songs in a playlist to be ranked higher than those that are not.
An equivalence between bipartite ranking and binary classification further enables efficient 
approximation of optimal parameters.
Empirical evaluations on two real playlist datasets demonstrate the effectiveness of the proposed method 
for cold-start playlist recommendation.
For future work, we would like to explore 
auxiliary data sources (\eg music information shared on social media) and additional features of songs and users 
(\eg lyrics, user profiles) % as well as the sequential order of songs in playlist %which could 
to make better recommendations.
Further, non-linear models such as deep neural networks have been shown to work extremely well in a wide range of tasks,
and the proposed linear model with sparse parameters %in this work 
could be more compact if a non-linear model %objective 
were adopted.

\clearpage
\onecolumn
\appendix
\begin{center}
  {\Large\bf Appendix to ``Cold-start Playlist Recommendation with Multitask Learning''}
\end{center}
\rule{0pt}{50pt}

\section{Proof of Lemma~\ref{lm:rank2clf}}

First, we can approximate the empirical risk $R_{\uptheta}^{\textsc{rank}}$ (with the exponential surrogate) as follows:
\begin{equation*}
\begin{aligned}
R_{\uptheta}^{\textsc{rank}}(f, \DCal)
&= \frac{1}{N} \sum_{u=1}^U \sum_{i \in P_u} \frac{1}{M_-^i} \sum_{m': y_{m'}^i = 0} \exp \left( -\min_{m: y_m^i = 1} f(m, u, i) + f(m', u, i) \right) \\
&= \frac{1}{N} \sum_{u=1}^U \sum_{i \in P_u} \frac{1}{M_-^i} \exp \left( -\min_{m: y_m^i = 1} f(m, u, i) \right) 
   \sum_{m': y_{m'}^i = 0} \exp \left( f(m', u, i) \right) \\
&\approx \frac{1}{N} \sum_{u=1}^U \sum_{i \in P_u} \frac{1}{M_-^i} \exp \left( \frac{1}{p} \log \sum_{m: y_m^i = 1} e^{-p f(m, u, i)} \right)
   \sum_{m': y_{m'}^i = 0} \exp \left( f(m', u, i) \right) \\
&= \frac{1}{N} \sum_{u=1}^U \sum_{i \in P_u} \frac{1}{M_-^i} \left( \sum_{m: y_m^i = 1} e^{-p f(m, u, i)} \right)^\frac{1}{p} 
   \sum_{m': y_{m'}^i = 0} e^{f(m', u, i)} \\
&= \frac{1}{N} \sum_{u=1}^U \sum_{i \in P_u} \frac{1}{M_-^i} \left( \left( \sum_{m': y_{m'}^i = 0} e^{f(m', u, i)} \right)^p 
   \sum_{m: y_m^i = 1} e^{-p f(m, u, i)} \right)^\frac{1}{p} \\
&= \frac{1}{N} \sum_{u=1}^U \sum_{i \in P_u} \frac{1}{M_-^i} \left( 
   \sum_{m: y_m^i = 1} e^{-p f(m, u, i)} \left( \sum_{m': y_{m'}^i = 0} e^{f(m', u, i)} \right)^p \right)^\frac{1}{p} \\
&= \frac{1}{N} \sum_{u=1}^U \sum_{i \in P_u} \frac{1}{M_-^i} \left( 
   \sum_{m: y_m^i = 1} \left( \sum_{m': y_{m'}^i = 0} e^{- \left( f(m, u, i) - f(m', u, i) \right)} \right)^p \right)^\frac{1}{p} \\
&= \widetilde R_{\uptheta}^{\textsc{rank}}(f, \DCal).
\end{aligned}
\end{equation*}

Recall that $R_{\uptheta}^\textsc{mtc}$ is the following classification risk:
\begin{equation*}
R_{\uptheta}^{\textsc{mtc}}(f, \DCal)
= \frac{1}{N} \sum_{u=1}^U \sum_{i \in P_u} \left( 
  \frac{1}{p M_+^i} \sum_{m: y_m^i = 1} e^{-p f(m, u, i)} 
  + \frac{1}{M_-^i} \sum_{m': y_{m'}^i = 0} e^{f(m', u, i)} \right),
\end{equation*}

Let $\uptheta^* \in \argmin_{\uptheta} R_{\uptheta}^{\textsc{mtc}}$ (assuming minimisers exist),
we want to prove that $\uptheta^* \in \argmin_{\uptheta} \widetilde R_{\uptheta}^{\textsc{rank}}$.

\begin{proof}
We follow the proof technique in~\cite{ertekin2011equivalence}
by first introducing a constant feature $1$ for each song,
without loss of generality, let the first feature of $\x_m, \, m \in \{1,\dots,M\}$ be the constant feature, \ie $x_m^0 = 1$.
We can show that
$\frac{\partial \, R_{\uptheta}^{\textsc{mtc}}} {\partial \, \uptheta} = 0$ implies
$\frac{\partial \, \widetilde R_{\uptheta}^{\textsc{rank}}} {\partial \, \uptheta} = 0$,
which means minimisers of $R_{\uptheta}^{\textsc{mtc}}$ also minimise $\widetilde R_{\uptheta}^{\textsc{rank}}$.

Let 
%\begin{equation*}
%\begin{aligned}
$
\displaystyle
0 
= \frac{\partial \, R_{\uptheta}^{\textsc{mtc}}} {\partial \, \beta_i^0}
= \frac{1}{N} \left( 
   \frac{1}{p M_+^i} \sum_{m: y_m^i = 1} e^{-p f(m, u, i)} (-p)
   + \frac{1}{M_-^i} \sum_{m': y_{m'}^i = 0} e^{f(m', u, i)} \right),
\ \forall i \in P_u, \, u \in \{1,\dots,U\},
$
%\end{aligned}
%\end{equation*}

we have
\begin{equation}
\label{eq:eq1}
\frac{1}{M_+^i} \sum_{m: y_m^i = 1} e^{-p f(m, u, i)} \Bigg|_{\uptheta = \uptheta^*}
= \frac{1}{M_-^i} \sum_{m': y_{m'}^i = 0} e^{f(m', u, i)} \Bigg|_{\uptheta = \uptheta^*}, 
\ \forall i \in P_u, \, u \in \{1,\dots,U\},
\end{equation}

Further, let
\begin{equation*}
\zero 
= \frac{\partial \, R_{\uptheta}^{\textsc{mtc}}} {\partial \, \betabm_i} 
= \frac{1}{N} \left( 
   \frac{1}{p M_+^i} \sum_{m: y_m^i = 1} e^{-p f(m, u, i)} (-p \x_m)
   + \frac{1}{M_-^i} \sum_{m': y_{m'}^i = 0} e^{f(m', u, i)} \x_{m'} \right),
\ \forall i \in P_u, \, u \in \{1,\dots,U\},
\end{equation*}
we have
\begin{equation}
\label{eq:eq2}
\frac{1}{M_+^i} \sum_{m: y_m^i = 1} e^{-p f(m, u, i)} \x_m \Bigg|_{\uptheta = \uptheta^*}
= \frac{1}{M_-^i} \sum_{m': y_{m'}^i = 0} e^{f(m', u, i)} \x_{m'} \Bigg|_{\uptheta = \uptheta^*},
\ \forall i \in P_u, \, u \in \{1,\dots,U\}.
\end{equation}

Note that $\forall i \in P_u, \, u \in \{1,\dots,U\}$,
\begin{equation}
\label{eq:eq3}
\resizebox{\textwidth}{!}{$
\begin{aligned}
&\frac{\partial \, \widetilde R_{\uptheta}^{\textsc{rank}}} {\partial \, \betabm_i} \Bigg|_{\uptheta = \uptheta^*} \\
&= \frac{1}{N M_-^i} \left[
   \frac{1}{p} \left( \sum_{m: y_m^i = 1} e^{-p f(m, u, i)} \right)^{\frac{1}{p} - 1} 
   \sum_{m: y_m^i = 1} e^{-p f(m, u, i)} (-p \x_m) 
   \sum_{m': y_{m'}^i = 0} e^{f(m', u, i)} + 
   \left( \sum_{m: y_m^i = 1} e^{-p f(m, u, i)} \right)^\frac{1}{p}
   \sum_{m': y_{m'}^i = 0} e^{f(m', u, i)} \x_{m'} \right] \\
&= \frac{-1}{N M_-^i}
   \left( \sum_{m: y_m^i = 1} e^{-p f(m, u, i)} \right)^{\frac{1}{p} - 1} 
   \left[
   \sum_{m: y_m^i = 1} e^{-p f(m, u, i)} \x_m 
   \sum_{m': y_{m'}^i = 0} e^{f(m', u, i)} -
   \sum_{m: y_m^i = 1} e^{-p f(m, u, i)}
   \sum_{m': y_{m'}^i = 0} e^{f(m', u, i)} \x_{m'} \right] \\
&= \frac{-1}{N M_-^i}
   \left( \sum_{m: y_m^i = 1} e^{-p f(m, u, i)} \right)^{\frac{1}{p} - 1} 
   \left[
   \left( \sum_{m: y_m^i = 1} e^{-p f(m, u, i)} \x_m \right)
   \left( \frac{M_-^i}{M_+^i} \sum_{m: y_m^i = 1} e^{-p f(m, u, i)} \right) -
   \sum_{m: y_m^i = 1} e^{-p f(m, u, i)}
   \sum_{m': y_{m'}^i = 0} e^{f(m', u, i)} \x_{m'} \right] \\
&  \hspace{1.5em} (\text{by Eq.~\ref{eq:eq1}}) \\
&= \frac{-1}{N M_-^i}
   \left( \sum_{m: y_m^i = 1} e^{-p f(m, u, i)} \right)^\frac{1}{p}
   \left[
   \frac{M_-^i}{M_+^i}
   \sum_{m: y_m^i = 1} e^{-p f(m, u, i)} \x_m -
   \sum_{m': y_{m'}^i = 0} e^{f(m', u, i)} \x_{m'} \right] \\
&= \zero \ (\text{by Eq.~\ref{eq:eq2}}).
\end{aligned}
$}
\end{equation}

Let 
\begin{equation*}
h(u, i) 
= \frac{1}{N M_-^i} \left( \sum_{m: y_m^i = 1} e^{-p f(m, u, i)} \right)^\frac{1}{p} 
  \sum_{m': y_{m'}^i = 0} e^{f(m', u, i)},
\ \forall i \in P_u, \, u \in \{1,\dots,U\}.
\end{equation*}

Similar to Eq.~(\ref{eq:eq3}), we have
\begin{equation}
\label{eq:eq4}
\frac{\partial \, h(u, i)}{\partial \, \betabm_i} \Bigg|_{\uptheta = \uptheta^*} = \zero,
\ \forall i \in P_u, \, u \in \{1,\dots,U\}.
\end{equation}

Note that $\forall u \in \{1,\dots,U\}$,
by Eq.~(\ref{eq:eq4})
\begin{equation}
\label{eq:eq5}
\frac{\partial \, \widetilde R_{\uptheta}^{\textsc{rank}}} {\partial \, \alphabm_u} \Bigg|_{\uptheta = \uptheta^*} 
= \sum_{i \in P_u} \frac{\partial \, h(u, i)}{\partial \, \alphabm_u}  \Bigg|_{\uptheta = \uptheta^*}
= \sum_{i \in P_u} \frac{\partial \, h(u, i)}{\partial \, \betabm_i} \Bigg|_{\uptheta = \uptheta^*}
= \zero,
\end{equation}
and
\begin{equation}
\label{eq:eq6}
\frac{\partial \, \widetilde R_{\uptheta}^{\textsc{rank}}} {\partial \, \bmu} \Bigg|_{\uptheta = \uptheta^*} 
= \sum_{u=1}^U \sum_{i \in P_u} \frac{\partial \, h(u, i)}{\partial \, \bmu} \Bigg|_{\uptheta = \uptheta^*}
= \sum_{u=1}^U \sum_{i \in P_u} \frac{\partial \, h(u, i)}{\partial \, \betabm_i} \Bigg|_{\uptheta = \uptheta^*}
= \zero.
\end{equation}

Finally, by Eq.~(\ref{eq:eq3}), Eq.~(\ref{eq:eq5}), and Eq.~(\ref{eq:eq6}), $\uptheta^* \in \argmin_{\uptheta} \widetilde R_{\uptheta}^{\textsc{rank}}$.

\end{proof}

\clearpage
\newpage

\section{Evaluation metrics}
The four evaluation metrics used in this work are:
\begin{itemize}
\item \emph{HitRate@K}, which is also known as Recall@K, is the number of correctly recommended songs amongst the top-$K$ recommendations over
      the number of songs in the observed playlist.
\item \emph{Area under the ROC curve} (AUC), which is the probability that a positive instance is ranked higher than a negative instance (on average).
\item \emph{Novelty} measures the ability of a recommender system to suggest previously unknown (\ie novel) items,
      $$
      \text{Novelty@K} 
      = \frac{1}{U} \sum_{u=1}^U \frac{1}{|P_u^\textsc{test}|} \sum_{i \in P_u^\textsc{test}} \sum_{m \in S_K^i} 
        \frac{-\log_2 pop_m}{K},
      $$
      where $P_u^\textsc{test}$ is the (indices of) test playlists from user $u$, 
      $S_K^i$ is the set of top-$K$ recommendations for test playlist $i$ and $pop_m$ is the popularity of song $m$.
      Intuitively, the more popular a song is, the more likely a user is to be familiar with it, and therefore the less likely to be novel.
\item \emph{Spread} measures the ability of a recommender system to spread its attention across all possible items.
      It is defined as the entropy of the distribution of all songs,
      $$
      \text{Spread} = -\sum_{m=1}^M P(m) \log P(m),
      $$
      where $P(m)$ denotes the probability of song $m$ being recommended,
      which is computed from the scores of all possible songs using the \emph{softmax} function in this work.
\end{itemize}

\clearpage
\newpage
\twocolumn

% span two columns
\twocolumn[
\section{Dataset}
\rule{0pt}{10pt}  % gap
]

\begin{figure}[!h]
    \centering
    \includegraphics[width=.98\columnwidth]{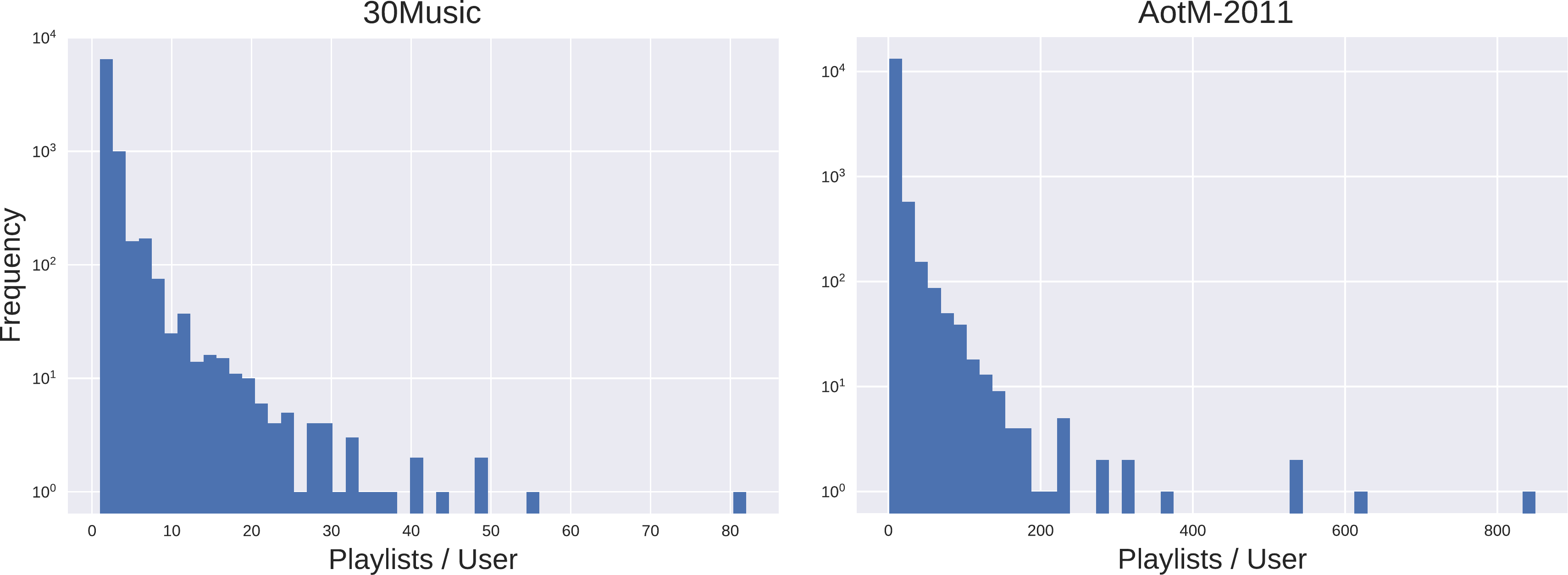}
    \caption{Histogram of the number of playlists per user}
    \label{fig:hist_pluser}
\end{figure}

\begin{figure}[!h]
    \centering
    \includegraphics[width=.98\columnwidth]{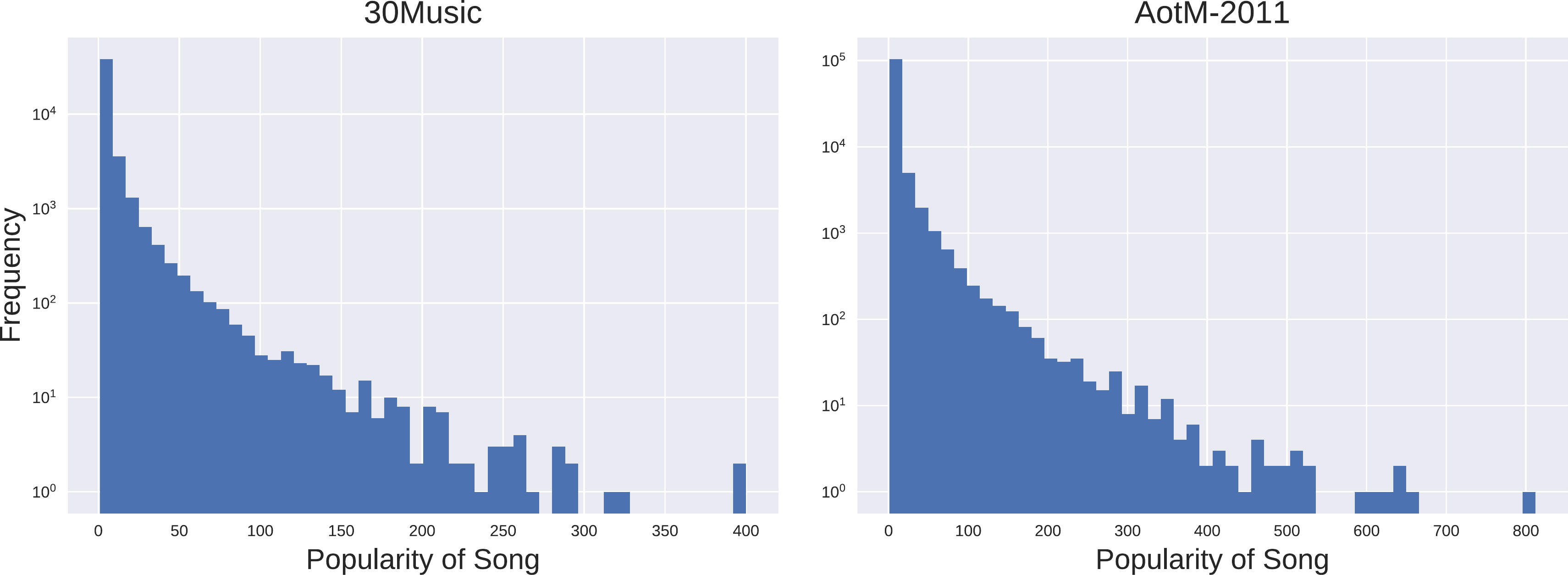}
    \caption{Histogram of song popularity}
    \label{fig:hist_songpop}
\end{figure}

The histograms of the number of playlists per user as well as song popularity 
%(\ie the number of occurrences of a song in all playlists)
of the two datasets are shown in Figure~\ref{fig:hist_pluser} and Figure~\ref{fig:hist_songpop},
respectively.
We can see from Figure~\ref{fig:hist_pluser} and Figure~\ref{fig:hist_songpop} that both the number
of playlists per user and song popularity follow a long-tailed distribution, which imposes further challenge to the learning task 
as the amount of data is very limited for users (or songs) at the tail.

\newpage

The training and test split of the two playlist datasets in the three cold-start settings are shown in 
Table~\ref{tab:stats3}, Table~\ref{tab:stats4}, and Table~\ref{tab:stats1}, respectively.

\begin{table}[!h]
    \centering
    \caption{Dataset for \emph{cold playlists}}
    \label{tab:stats3}
    \begin{tabular}{lrrcrr}
        \toprule
        \multirow{2}{*}{Dataset}  & \multicolumn{2}{c}{Training Set} && \multicolumn{2}{c}{Test Set} \\ \cmidrule{2-3} \cmidrule{5-6}
                                  & Playlists & Users && Playlists & Users \\
        \midrule
        30Music   & 15,262 &  8,070 && 2,195  & 1,644 \\
        AotM-2011 & 75,477 & 14,182 && 9,233  & 2,722 \\
        \bottomrule
    \end{tabular}
\end{table}

\begin{table}[!h]
    \centering
    \caption{Dataset for \emph{cold users}}
    \label{tab:stats4}
    \begin{tabular}{lrrcrr}
        \toprule
        \multirow{2}{*}{Dataset}  & \multicolumn{2}{c}{Training Set} && \multicolumn{2}{c}{Test Set} \\ \cmidrule{2-3} \cmidrule{5-6}
                                  & Users & Playlists && Users & Playlists \\
        \midrule
        30Music   & 5,649 & 14,067 && 2,420 & 3,390 \\
        AotM-2011 & 9,928 & 76,450 && 4,254 & 8,260 \\
        \bottomrule
    \end{tabular}
\end{table}

\begin{table}[!h]
    \centering
    \caption{Dataset for \emph{cold songs}}
    \label{tab:stats1}
    \begin{tabular}{lrrcrr}
        \toprule
        \multirow{2}{*}{Dataset}  & \multicolumn{2}{c}{Training Set} && \multicolumn{2}{c}{Test Set} \\ \cmidrule{2-3} \cmidrule{5-6}
                                  & Songs & Playlists && Songs & Playlists \\
        \midrule
        30Music   & 40,468  & 17,342 && 5,000  & 8,215 \\
        AotM-2011 & 104,428 & 84,646 && 10,000 & 19,504 \\
        \bottomrule
    \end{tabular}
\end{table}

\clearpage
\newpage
\onecolumn

% span two columns
%\twocolumn[
\section{Empirical results}
%\rule{0pt}{10pt}  % gap
%]

\begin{figure*}[!h]
    \centering
    \begin{minipage}{0.45\textwidth}
        \centering
        \includegraphics[width=\linewidth]{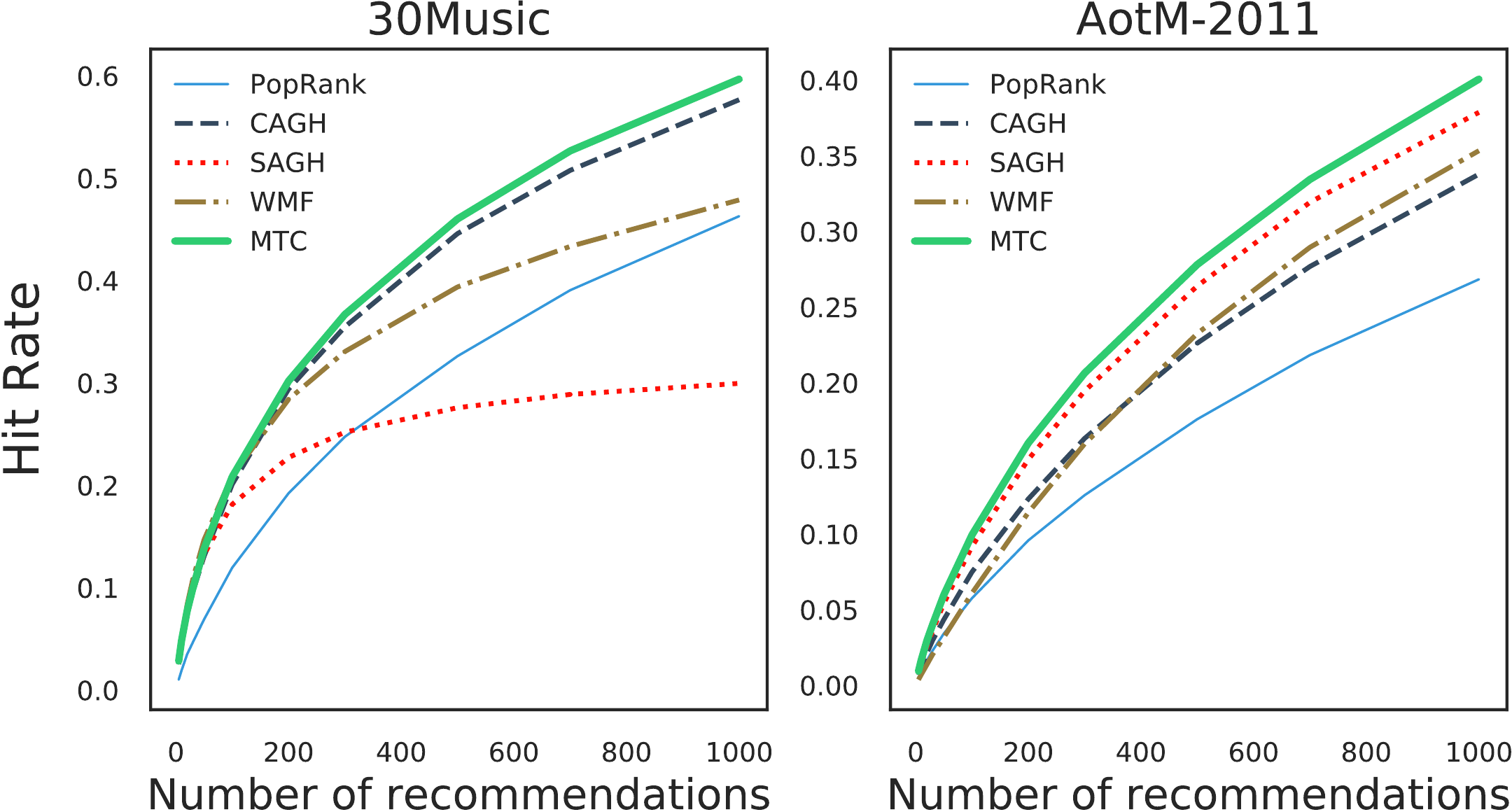}
        \caption{Hit rate of recommendation in the \emph{cold playlists} setting. \emph{Higher} values indicate better performance.}
        \label{fig:hr3}
    \end{minipage}\hspace{15pt}
    \begin{minipage}{0.45\textwidth}
        \centering
        \includegraphics[width=\linewidth]{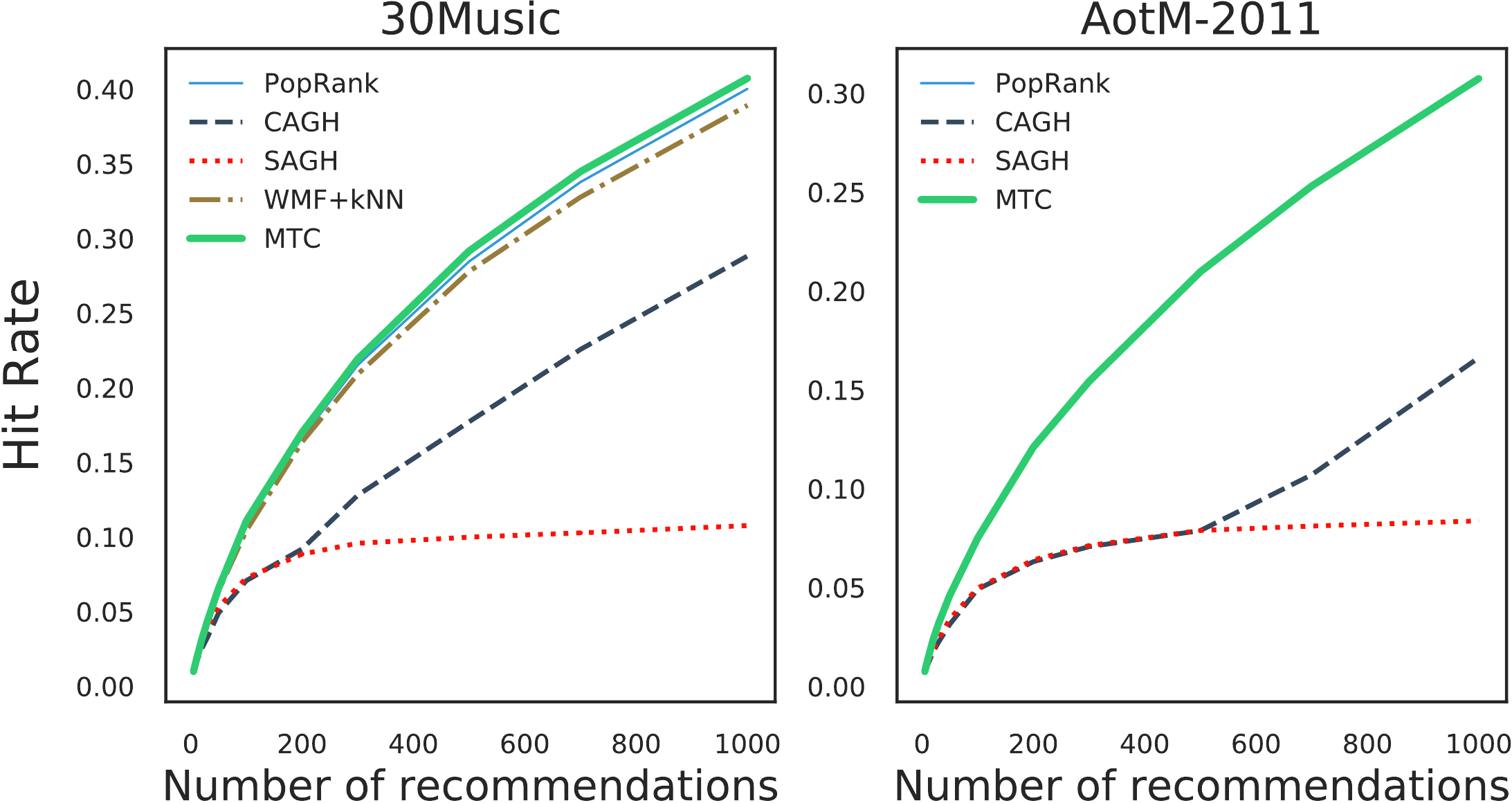}
        \caption{Hit rate of recommendation in the \emph{cold users} setting. \emph{Higher} values indicate better performance.}
        \label{fig:hr4}
    \end{minipage}
\end{figure*}

\begin{figure*}[!h]
    \centering
    \begin{minipage}{0.45\textwidth}
        \centering
        \includegraphics[width=\columnwidth]{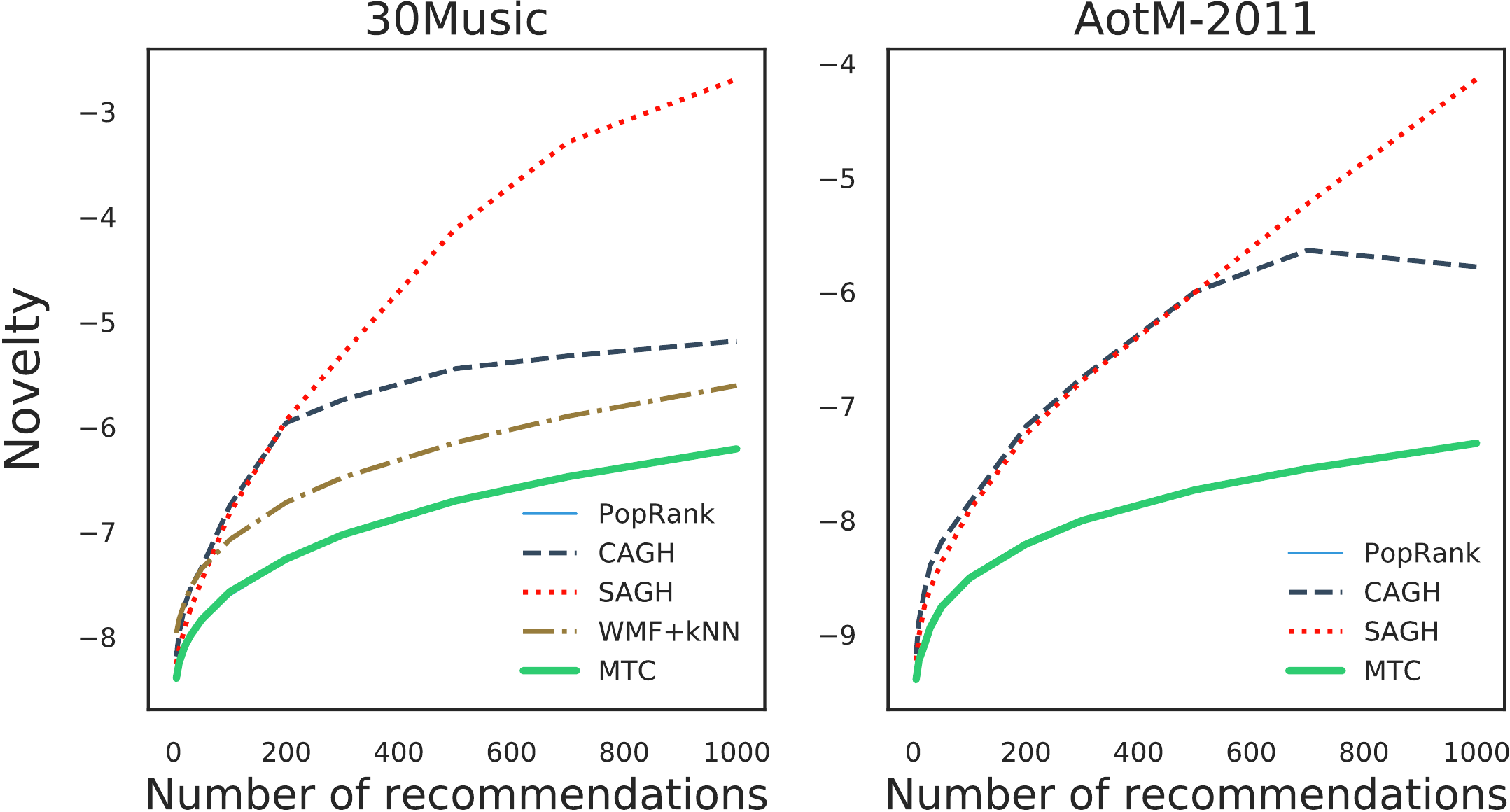}
        \caption{Novelty of recommendation in the \emph{cold users} setting. \emph{Moderate} values are preferable.}
        \label{fig:nov4}
    \end{minipage}\hspace{15pt}
    \begin{minipage}{0.45\textwidth}
        \centering
        \includegraphics[width=\columnwidth]{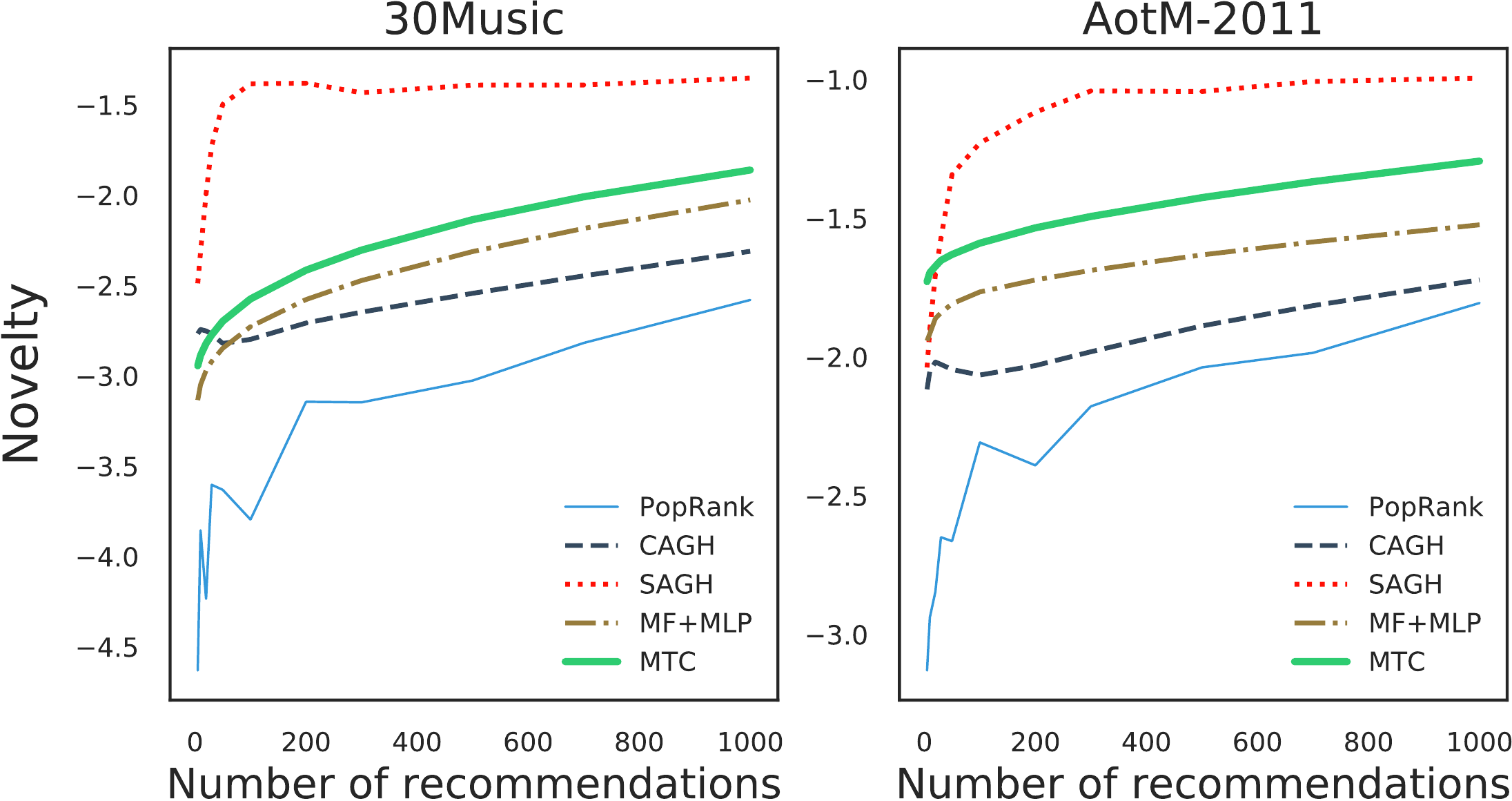}
        \caption{Novelty of recommendation in the \emph{cold songs} setting. \emph{Moderate} values are preferable.}
        \label{fig:nov1}
    \end{minipage}
\end{figure*}

\clearpage
\newpage
\onecolumn

%\section{Notations}
\section{Notations}

We introduce notations in Table~\ref{tab:notation}.
\begin{table}[!h]
\caption{Notations used in this paper}
\label{tab:notation}
\renewcommand{\arraystretch}{1.5} % tweak the space between rows
\setlength{\tabcolsep}{1pt} % tweak the space between columns
\centering
\begin{tabular}{llll}
\toprule
\multicolumn{3}{l}{\textbf{Notation}} & \textbf{Description} \\ \midrule
$D$        &  $\in$  &  $\Z^+ \rule{10pt}{0pt}$  & The number of features for each song \\
$M$        &  $\in$  &  $\Z^+$            & The number of songs, indexed by $m, m' \in \{1,\dots,M\}$ \\
$M_+^i$    &  $\in$  &  $\Z^+$            & The number of songs in playlist $i$ \\
$M_-^i$    &  $\in$  &  $\Z^+$            & The number of songs not in playlist $i$, \ie $M_-^i = M - M_+^i$ \\
$N$        &  $\in$  &  $\Z^+$            & The total number of playlists from all users \\
$U$        &  $\in$  &  $\Z^+$            & The number of users, indexed by $u \in \{1,\dots,U\}$ \\
$P_u$      &         &                    & The set of indices of playlists from user $u$ \\
$\alphabm_u$   &  $\in$  &  $\R^D$        & The weights of user $u$ \\
$\betabm_i$  &  $\in$  &  $\R^D$          & The weights of playlist $i$ from user $u$, $i \in P_u$ \\
$\bmu$     &  $\in$  &  $\R^D$            & The weights shared by all users (and playlists) \\
$\w_{u,i}$ &  $\in$  &  $\R^D$            & The weights of playlist $i$ from user $u$, $\w_{u,i} = \alphabm_u + \betabm_i + \bmu$ \\
%$\Y$       &  $\in$  &  $\R^{M \times N}$ & The matrix of binary labels that indicating if a song is in a playlist \\
$y_m^i$    &  $\in$  &  $\R^D$            & The positive binary label $y_m^i = 1$, \ie song $m$ is in playlist $i$ \\
$y_{m'}^i$ &  $\in$  &  $\R^D$            & The negative binary label $y_{m'}^i = 0$, \ie song $m'$ is not in playlist $i$ \\
%$\X$       &  $\in$  &  $\R^{M \times D}$ & The matrix of features of all songs \\
$\x_m$     &  $\in$  &  $\R^D$            & The feature vector of song $m$ \\
%$\x_{m'}$  &  $\in$  &  $\R^D$            & The feature vector of song $m'$ \\
\bottomrule
\end{tabular}
\end{table}

\end{document}